\documentclass[letterpaper, 10 pt, conference]{ieeeconf}  \renewcommand{\baselinestretch}{0.91}

\IEEEoverridecommandlockouts                             
\usepackage{float}
\usepackage{verbatim}
\usepackage{paralist}
\usepackage{amsfonts}
\usepackage{graphics} 
\usepackage{epsfig} 
\usepackage{mathptmx} 
\usepackage{times} 
\usepackage{amsmath} 
\usepackage{amssymb}  
\usepackage{physics}
\usepackage{mdwmath}
\usepackage{mdwtab}
\usepackage{color}
\usepackage[hidelinks]{hyperref}
\usepackage{algpseudocode}
\usepackage{algorithm}
\usepackage{caption}
\usepackage{subcaption}
\usepackage{amsmath} 
\usepackage{bm}      

\usepackage{tikz}
\usetikzlibrary{shapes, arrows.meta, positioning}

\newtheorem{theorem}{Theorem}%
\newtheorem{lemma}{Lemma}%
\newtheorem{proposition}{Proposition}%
\newtheorem{definition}{Definition}
\newtheorem{remark}{Remark}

\newtheorem{corollary}{Corollary}
\newtheorem{assumption}{Assumption}
\def\BibTeX{{\rm B\kern-.05em{\sc i\kern-.025em b}\kern-.08em
    T\kern-.1667em\lower.7ex\hbox{E}\kern-.125emX}}
\begin{document}

\title{
Quantum-Assisted Barrier Sequential Quadratic Programming for Nonlinear Optimal Control
}

\author{Nahid Binandeh Dehaghani, Rafal Wisniewski, A. Pedro Aguiar 
\thanks{N. Dehaghani and R. Wisniewski are with Department of Electronic Systems, Aalborg University, Fredrik Bajers vej 7c, DK-9220 Aalborg, Denmark
 {\tt\small nahidbd@es.aau.dk, raf@es.aau.dk}}
\thanks{A. Aguiar is with the Research Center for Systems and Technologies (SYSTEC-ARISE), Faculty of Engineering, University of Porto, Rua Dr. Roberto Frias sn, i219, 4200-465 Porto, Portugal
        {\tt\small pedro.aguiar@fe.up.pt}}%
 }

\maketitle

\begin{abstract}
We propose a quantum-assisted framework for solving constrained finite-horizon nonlinear optimal control problems using a barrier Sequential Quadratic Programming (SQP) approach. 
Within this framework, a quantum subroutine is incorporated to efficiently solve the Schur complement step using block-encoding and Quantum Singular Value Transformation (QSVT) techniques.
We formally analyze the time complexity and convergence behavior under the cumulative effect of quantum errors, establishing local input-to-state stability and convergence to a neighborhood of the stationary point, with explicit error bounds in terms of the barrier parameter and quantum solver accuracy.
The proposed framework  enables computational complexity to scale polylogarithmically with the system dimension
demonstrating the potential of quantum algorithms to enhance classical optimization routines in nonlinear control applications.
\end{abstract}


\section{Introduction}
Nonlinear optimal control problems (OCPs) are central to a wide range of engineering applications \cite{rawlings2020model,grune2016nonlinear}.
A common approach is to apply direct methods, by discretizing the problem and transforming it into a finite-dimensional nonlinear program (NLP).
To solve such NLPs, sequential quadratic programming (SQP) is widely adopted due to its strong convergence properties and its ability to handle general nonlinear constraints.
However, SQP relies on solving quadratic programming (QP) subproblems at each iteration, and the computational cost of these subproblems grows rapidly with problem size, limiting its applicability in large-scale and real-time control scenarios.

Recent advances in quantum algorithms for linear algebra, particularly block encoding and quantum singular value transformation (QSVT) \cite{gilyen2019quantum}, offer the potential to accelerate the solution of structured linear systems, eigenvalue problems, and optimization subroutines. 
While quantum methods have been explored for certain classes of optimization problems \cite{farhi2014quantum, childs2017quantum, harrow2009quantum,dehaghani2025LQG}, their integration into iterative nonlinear optimization frameworks—especially when incorporating uncertainty and robustness considerations to ensure that algorithmic guarantees persist under the cumulative effect of quantum errors—remains largely unexplored. 
The key contributions are: (i) the development of a barrier SQP algorithm with quantum subroutines while preserving convergence guarantees; (ii) the use of block-encoding and QSVT-based inversion to efficiently solve the Schur complement systems; 
(iii) complexity and error analysis, including input-to-state stability guarantees; (iv) validation on the Perelson HIV-1 model \cite{perelson1999}.


\section{Finite-Horizon Nonlinear OCP}
Consider the discrete-time nonlinear dynamical system 
    $$x_{k+1} = f(x_k, u_k), \quad k = 0, 1, \dots, N-1,$$
where \( x_k \in \mathbb{R}^n \) denotes the state, \( u_k \in \mathbb{R}^m \) the control input, and \( f: \mathbb{R}^n \times \mathbb{R}^m \rightarrow \mathbb{R}^n \) the system dynamics.  
We assume that \( f \) is Lipschitz in \(x\), and continuously differentiable in both \(x\) and \(u\). 
The initial state is fixed as \( x_0 = x_{\mathrm{init}} \).
The finite-horizon optimal control problem is to determine sequences 
\(\{x_k, u_k\}_{k=0}^{N-1}\) that minimize a cumulative cost functional subject to dynamics and constraints:
\begin{equation}\label{eq:OCP}
\begin{aligned}
    \min_{\{x_k, u_k\}} \quad & J = \sum_{k=0}^{N-1} \ell(x_k, u_k) + \ell_N(x_N) \\
    \text{s.t.} \quad & x_{k+1} = f(x_k, u_k), \quad k = 0, \dots, N-1, \\
    & c(x_k, u_k) \leq 0, \quad k = 0, \dots, N-1, \\
    & c_N(x_N) \leq 0, \\
    & x_0 = x_{\mathrm{init}},
\end{aligned}
\end{equation}
where \(\ell: \mathbb{R}^n \times \mathbb{R}^m \to \mathbb{R}\) denotes the stage cost, 
\(\ell_N: \mathbb{R}^n \to \mathbb{R}\) the terminal cost, 
and \(c:\mathbb{R}^n\times\mathbb{R}^m\to\mathbb{R}^p\), \(c_N:\mathbb{R}^n\to\mathbb{R}^{p_N}\) represent path and terminal inequality constraints, respectively. 
We assume that the cost function is bounded below, and that both the cost and constraint functions are continuously differentiable. Furthermore, we assume that the optimal control problem admits at least one feasible solution.
A broad class of numerical algorithms has been developed for solving~\eqref{eq:OCP}.  
Most methods rely on direct transcription, reformulating the OCP as a nonlinear program (NLP)~\cite{rawlings2020model,grune2016nonlinear}.

\paragraph*{NLP formulation} 
To solve the optimal control problem~\eqref{eq:OCP},
we adopt a direct multiple-shooting transcription, which treats all $\{x_k,u_k\}$ as decision variables and imposes the dynamics as equality constraints. This improves numerical robustness and enables parallel evaluation of dynamics and sensitivities.  
We define the stacked decision vector for each stage as
$z_k := \begin{bmatrix} x_k^T & u_k^T \end{bmatrix}^T \in \mathbb{R}^{n+m}, 
\quad k = 0,\dots,N-1$,
and the terminal variable as \(z_N := x_N\in \mathbb{R}^n\).  
The complete decision vector is then
$z := \begin{bmatrix} z_0^T & \cdots & z_{N-1}^T & z_N^T \end{bmatrix}^T 
\in \mathbb{R}^{n_z}$.
where $n_z=N(n+m)+n$. The NLP reads
\begin{equation} \label{eq:NLP}
\begin{aligned}
\min_{z} \quad 
& F(z) := \sum_{k=0}^{N-1} \ell(z_k) + \ell_N(z_N) \\
\text{s.t.}\quad 
& G(z) = 0,\qquad H(z) \le 0,
\end{aligned}    
\end{equation}
where
$G(z)$ and $H(z)$ are adequate vectors.


In our approach, we employ a barrier SQP method. 
The inequality constraints are incorporated through a barrier term, leading to the barrier-augmented objective
$\tilde{F}(z;\mu)\;=\;F(z)\;+\;\mu \sum_{j=1}^{q} \phi\!\big(H_j(z)\big)$,
where $\mu>0$ is the barrier parameter, $q=Np+p_N$ is the total number of 
inequality constraints, and 
 $\phi : (-\infty, 0) \to \mathbb{R}$ is a barrier function that is strictly convex and twice continuously differentiable on its domain. Moreover, it satisfies $\lim_{s \to 0^-} \phi(s) = +\infty$, with $\phi'(s) < 0$ and $\phi''(s) > 0$ for all $s < 0$. A typical example is the logarithmic barrier $\phi(s) = -\log(-s)$.
This reformulation approximately formulate the inequality constrained NLP problem as an equality constrained problem.
A general procedure to solve the resulting optimization problem is as follows 
\cite{wright1997primal,boyd2004convex}:
Starting with an initial guess \(z^{(0)}\) and barrier parameter \(\mu = \mu_0\) at iteration \(i=0\), the method proceeds iteratively through three main steps until convergence.

\paragraph*{Step 1: Linearization and Quadratic Approximation} 
At iteration \(i\), the cost, constraints, and dynamics are locally approximated around the current nominal trajectory \(z^{(i)}\).
Specifically, nonlinear equality constraints are linearized via first-order Taylor expansions, and the barrier-augmented objective is approximated by a second-order Taylor expansion.  
Introducing 
deviation variable
    $\Delta z^{(i)} = z - z^{(i)}$,
the resulting subproblem takes the form of a quadratic program (QP)
\begin{align*}
    \min_{\Delta z} \quad  \tfrac{1}{2} \Delta z^T Q \,\Delta z + \Delta z^Tg \quad
      \text{  s.t.} \quad  \mathcal{A} \,\Delta z = r,
\end{align*}
with $Q =\nabla^2 F\!(z^{(i)})
   + \mu \sum_{j=1}^{q} \phi''\!\big(H_j(z^{(i)})\big)\,\nabla H_j\!(z^{(i)})\nabla H_j\!(z^{(i)})^T,$
   $g = \nabla F\!\left(z^{(i)}\right)
   + \mu \sum_{j=1}^{q} \phi'\!\big(H_j(z^{(i)})\big)\,\nabla H_j\!\left(z^{(i)}\right)$, $r = -G(z^{(i)})$,
$\mathcal{A} = \nabla_z G(z^{(i)})$.
Here, \(Q\) is the barrier-augmented Gauss--Newton Hessian of the Lagrangian, \(g\) the barrier-augmented gradient, 
and \(\mathcal{A}\) the Jacobian of the equality constraints.

\paragraph*{Step 2: Solve the QP Subproblem}
Introducing the Lagrange multiplier $\lambda$ associated with the equality constraint, the quadratic subproblem can be written in the Karush–Kuhn–Tucker (KKT) form 
$ [
\begin{smallmatrix} 
Q & \mathcal{A}^T \\
\mathcal{A} & 0
\end{smallmatrix}]
[\begin{smallmatrix}
\Delta z \\[3pt] \lambda
\end{smallmatrix}]
=
[\begin{smallmatrix}
-g \\[3pt] r
\end{smallmatrix}]$.
Applying block elimination with the Schur complement reduces the system to $S \lambda= b$, $Q \Delta z = -\big(g+ \mathcal{A}^T \lambda\big)$, S = $\mathcal{A} Q^{-1} \mathcal{A}^T$,  
$b = -r - \mathcal{A} Q^{-1} g$,
where the Schur complement $S$ inherits important structural properties: it is symmetric positive semidefinite, and, due to the stagewise structure of \(\mathcal{A}\), it is sparse and block-banded.  
These properties can be exploited for efficient factorization and solution of the linear system. 

\paragraph*{Step 3: Update $z^{i+1}$ and convergence check}
After computing the step direction \(\Delta z\) from the quadratic subproblem, the nominal trajectory is updated according to
\[
z^{(i+1)} = z^{(i)} + \alpha \Delta z,
\]
where \(\alpha \in (0, 1]\) is a step size chosen to ensure convergence and feasibility.
To obtain $\alpha$, the step size is first capped using a fraction-to-the-boundary rule to reduce the number of backtracking steps
and the line search is initialized with $\alpha \gets \alpha_{\max}$. This guarantees that $H(z+\alpha\Delta z) < 0$ from the start, preserving strict feasibility.
Next, starting from this initial value, a backtracking line search is performed. The step size is iteratively reduced according to $\alpha \leftarrow \tau \alpha$, $\tau \in (0,1),$
until two conditions are satisfied. First, strict feasibility must hold:
$H\big(z^{(i)} + \alpha \Delta z\big) < 0$.
Second, the barrier-augmented objective must satisfy the Armijo condition:
$\tilde{F}\big(z^{(i)} + \alpha \Delta z;\mu\big)
\le\tilde{F}\big(z^{(i)};\mu\big) + c\,\alpha\, g^T \Delta z$,
where 
$c \in (0,1)$ is the Armijo parameter. This 
ensures 
each step remains feasible while achieving sufficient decrease in the objective.

Convergence of the barrier SQP method is declared if 
the optimality and feasibility tolerances are met.
If not, 
 the barrier parameter is updated according to a given rule\footnote{Unlike some classical barrier methods, the proposed algorithm does not employ inner loops, simplifying the overall structure. With an appropriate update rule, the classical approach with inner iterations can be recovered as a special case.}
$\mu^{(i+1)} = \mathcal{U}_\mu\!\big(\mu^{(i)},\,\cdot \big)$,
where $\mathcal{U}_\mu$ is a user-defined update function that may decrease $\mu$ (e.g., $\mathcal{U}_\mu = \beta \mu^{(i)}$, $\beta\in(0,1)$), keep it constant for several iterations, or adaptively adjust it based on progress measures (e.g., residual norms). 
Next, the algorithm proceeds with another iteration.

\begin{algorithm}[t]
\small
\caption{Hybrid Barrier SQP with Quantum Schur Step}
\label{alg:hybrid}
\begin{algorithmic}[1]
\Statex \textbf{Inputs:} initial feasible $z_{0}$ with $H(z_{0})<0$; initial barrier $\mu_0>0$;
Armijo $c\in(0,1)$; step tolerance $\varepsilon_{\Delta z}$; backtracking $\tau\in(0,1)$; tolerances $\varepsilon_{\mathrm{opt}}, \varepsilon_{\mathrm{feas}}, \mu_{\min}$.
\Statex \textbf{Output:} $z^\star$
\State Set $i \gets 0$, $z^{(i)} \gets z_{0}$, $\mu^{(i)} \gets \mu_0$.
\While{$\mu^{(i)} > \mu_{\min}$} \Comment{Barrier loop}
  \Statex \textbf{(Step 1) Linearization \& Quadratic Approximation at $z^{(i)}$:}
  \State Evaluate $G \gets G(z^{(i)})$, $H \gets H(z^{(i)})$, $\mathcal{A} \gets \nabla_z G(z^{(i)})$.
  \State Barrier-augmented gradient:
   $g \; \gets \; \nabla F\!\big(z^{(i)}\big)
   + \mu^{(i)} \sum_{j=1}^{q} \phi'\!\big(H_j(z^{(i)})\big)\,\nabla H_j\!(z^{(i)})$.
  \State Barrier (Gauss--Newton/stage-wise) Hessian:
$    Q \;\gets\; \nabla^2 F\!\left(z^{(i)}\right)
   + \mu^{(i)} \sum_{j=1}^{q} \phi''\!\big(H_j(z^{(i)})\big)\,\nabla H_j\!\big(z^{(i)}\big)\nabla H_j\!\big(z^{(i)}\big)^T$.
  \State Linearized equalities residual: $r \gets -\,G$.
  \vspace{2pt}

  \State \textbf{(Step 2) Quantum Schur step:}\\
$      (\Delta z) \;\gets\; \texttt{QuantumSchurStep}\big(Q,\mathcal{A},g,r; \varepsilon_{\Delta z}\big)$
  \emph{(This subroutine encapsulates $S=\mathcal{A}Q^{-1}\mathcal{A}^T$, $b=-r-\mathcal{A}Q^{-1}g$,
  solve $S\lambda=b$, and $\Delta z=-Q^{-1}(g+\mathcal{A}^T\lambda)$.)}
  \vspace{2pt}

\State \textbf{(Step 3) Fraction-to-the-boundary initialization \& backtracking:}
 $ \alpha_{\max} \;\gets\; \min\!\left(1,\; \theta \min_{j:\,\nabla H_j(z^{(i)})^T \Delta z > 0}
   \frac{-H_j(z^{(i)})}{\nabla H_j(z^{(i)})^T \Delta z}\right)$, $\theta \in (0,1)$
  \State Initialize $\alpha \gets \alpha_{\max}$.
  \While{$\exists j:\ H_j(z^{(i)}+\alpha \Delta z) \ge 0$
         \textbf{ or }
         $\tilde{F}(z^{(i)}+\alpha \Delta z;\mu) >
          \tilde{F}(z^{(i)};\mu) + c\,\alpha\, g^T \Delta z$}
     \State $\alpha \gets \tau\,\alpha$
  \EndWhile
  \State \textbf{Update iterate:} $z^{(i+1)} \gets z^{(i)} + \alpha\,\Delta z$.
  \vspace{2pt}

  \State \textbf{convergence check:}
    \If{$\|\nabla F(z^{(i+1)})\|\le \varepsilon_{\mathrm{opt}}$ \textbf{and} $\|G(z^{(i+1)})\|\le \varepsilon_{\mathrm{feas}}$}
      \State \textbf{break}
    \EndIf
  \State \textbf{Update barrier and continue:} $\mu^{(i+1)} \gets \mathcal{U}_\mu\!\big(\mu^{(i)}, \cdot\big)$, \, $i\gets i+1$
\EndWhile
\State \textbf{Return} $z^\star \gets z^{(i)}$
\end{algorithmic}
\end{algorithm}

\section{Hybrid Quantum--Classical SQP
}
This section describes how a quantum Schur complement solver is integrated into the classical barrier-SQP loop. The outer operations—barrier updates, evaluation of nonlinear functions and gradients, and backtracking line search—remain purely classical. The computational bottleneck, namely the linear algebra required to solve the KKT system, is delegated to a quantum subroutine that leverages block-encoding and QSVT-based primitives.
The complete hybrid loop is summarized in Algorithm~\ref{alg:hybrid}.

\begin{algorithm}[t]
\small
\caption{Quantum Schur Complement Step via Block Encodings and QSVT 
}
\label{Qschur}
\begin{algorithmic}[1]
\Statex \textbf{Input:}
Access to all given matrices and vectors $Q$, $\mathcal{A}$, $g$ $r$.
\Statex \textbf{Output:} Classical vector $\Delta z$
\Statex \textbf{Procedure:}
\State Prepare block encoding \( U_Q \) and \( U_{\mathcal{A}} \), \( U_{g} \), \( U_{r} \), with normalization factors $\alpha_{Q},\alpha_{\mathcal{A}}, \alpha_{g},\alpha_{r}$.
\Statex \textbf{Step 1: Form Schur System}
\State Invert \( Q \): Construct $ U_{Q^{-1}}  \leftarrow QSVT-Invert(U_Q)$
\State Form \( S \): Construct the Schur Complement: 
\( U_S \gets U_{\mathcal{A}} \cdot U_{Q^{-1}} \cdot U_{\mathcal{A}}^T \) 
\State Form the right-hand side using LCU: \( U_b \gets-U_r-U_{\mathcal{A}} U_{Q^{-1}} U_g \) 

\Statex \textbf{Step 2: Solve for Multipliers}
\State Invert Schur complement: $U_{S^{-1}} \leftarrow QSVT-Invert(U_S)$
\State Form $U_{\lambda} \leftarrow U_{S^{-1}} U_b$
\Statex \textbf{Step 3: Primal Step Recovery and Classical Readout}
\State Form intermediate encoding using LCU: $U_{1} \leftarrow U_g + U_{\mathcal{A}}^T U_{\lambda}$
\State Compute: $U_{\Delta z} \leftarrow - U_{Q^{-1}} U_{1}$ 
\State Conducting measurements on $U_{\Delta z}$ to obtain $\Delta z$
\end{algorithmic}
\end{algorithm}
\subsection{Quantum Assited framework for SQP }
We first introduce the concepts underlying Algorithm~\ref{Qschur}.

\begin{definition}[Block Encoding]
Let $A \in \mathbb{C}^{n \times n}$ act on $s$ qubits 
so that $n = 2^s$.
Suppose there exist a normalization factor $\alpha > 0$, error tolerance $\varepsilon \ge 0$, ancilla count $a \in \mathbb{N}$, and a unitary $U_A \in \mathbb{C}^{2^{s+a} \times 2^{s+a}}$ such that
$\Big\| A - \alpha \big( \langle 0^a | \otimes I_s \big) U_A \big( |0^a \rangle \otimes I_s \big) \Big\| \leq \varepsilon$,
then $U_A$ is an $(\alpha_A,a,\varepsilon)$-block encoding of $A$. 
\end{definition}

\begin{remark}[Matrix Padding for Block Encoding]
If $A \in \mathbb{C}^{m \times n}$ with $m,n \le 2^s$ for some integer $s$, then $A$ can be embedded into a square matrix $A_e \in \mathbb{C}^{2^s \times 2^s}$ by placing $A$ in the top-left $m \times n$ block and setting all other entries to zero. The padded matrix $A_e$ can then be block-encoded with respect to the full $2^s$-dimensional Hilbert space \cite{gilyen2019quantum}.
\end{remark}


\begin{remark}[Addition of Block Encodings] \label{rem:add}
Let \( U_A \) be an \( (\alpha_A, a, \epsilon_A) \)-block encoding of a matrix \( A \), and let \( U_B \) be a \( (\alpha_B, b, \epsilon_B) \)-block encoding of a matrix \( B \), both acting on \( s \) qubits. Then, there exists an \( (\alpha_A + \alpha_B, \max(a, b) + 1, \epsilon_A+\epsilon_B ) \)-block encoding of \( A + B \), achieved by defining the controlled unitary
$U = \ket{0}\bra{0} \otimes U_A + \ket{1}\bra{1} \otimes U_B$,
and applying it within the sandwich construction:
$(V^\dagger \otimes I)\, U\, (V \otimes I)$,
where $ V=\frac{1}{\sqrt{\alpha_A+ \alpha_B}} 
[\begin{smallmatrix} 
 \sqrt{\alpha_A} & \cdot \\
\sqrt{\alpha_B} & \cdot
\end{smallmatrix} ]$ is a unitary.
This construction is referred to as the linear combination of unitaries (LCU) 
\end{remark}
\begin{remark}[Multiplication of Block Encodings]
\label{rem:mul}
\noindent
Let $U_A$ and $U_B$ be block encodings of matrices $A$ and $B$, with parameters $(\alpha_A, a, \epsilon_A)$ and $(\alpha_B, b, \epsilon_B)$, respectively. Then the unitary $U = (I_b \otimes U_A)(I_a \otimes U_B)$
is an \( (\alpha_A \alpha_B, a + b, \alpha_A\epsilon_A + \alpha_B \epsilon_B) \)-block encoding of \( AB \), where the ancilla identity operators \( I_b \) and \( I_a \) act on registers of sizes \( b \) and \( a \) 
\cite{lin2022lecture,chakraborty2018power}.
\end{remark}
\begin{remark}[Matrix Inversion via QSVT] \label{rem:QSVT}
Consider matrix
$A \in \mathbb{C}^{n \times n}$ with singular value decomposition
$A = W \Sigma V^\dagger$,
where $\Sigma$ is diagonal with singular values $\sigma_i \in [\kappa^{-1},1]$, and $\kappa$ denotes the condition number of $A$. The inverse is given by
$A^{-1} = V \,\mathrm{diag}\!\big(f(\sigma_i)\big) W^\dagger$, with $f(x) = x^{-1}$.
The QSVT framework allows approximation of $f(x)$ on the spectrum of $A$ using polynomial transformations applied to a block encoding of $A$. Since $f(x)=1/x$ is odd, we approximate it by an odd polynomial $p(x)$ such that
$\Big|p(x) - \tfrac{1}{\kappa \beta}\tfrac{1}{x}\Big| \leq \varepsilon^\prime$, $\forall x \in \Big[-1,-\tfrac{1}{\kappa}\Big] \cup \Big[\tfrac{1}{\kappa},1\Big]$,
with the additional constraints $p(-x)=-p(x)$ (odd parity) and $|p(x)| \leq 1$ on the same domain. 
The factor $\beta$ is chosen to ensure that $|p(x)| \leq 1$ for all $x \in [-1,1]$.
The resulting block encoding of $A^{-1}$ has normalization factor
  $\alpha_{A^{-1}} \;=\; \frac{\kappa\,\beta}{\alpha_A}$,
and overall inversion error tolerance
  $\varepsilon_{A^{-1}}$ proportional to 
  $\varepsilon_A$ added with 
$\varepsilon^\prime$ \cite{gilyen2019quantum}.
\end{remark}

\subsubsection*{Quantum Schur Step Algorithmic Description}
Algorithm~\ref{Qschur} presents the quantum subroutine for solving the KKT system within the barrier-SQP iteration.  
It operates on block encodings of the problem data $(Q,\mathcal{A},g,r)$ with normalization factors 
$\|Q\| \leq \alpha_{Q}, \|\mathcal{A}\| \leq \alpha_{\mathcal{A}},
\|g\| \leq \alpha_{g},
\|r\| \leq \alpha_{r}$, and employs QSVT-based inversion for the required linear solves. 
The output is the classical search direction $\Delta z$, obtained by measuring the quantum state encoding the solution and postprocessing the result. 
Figure~\ref{fig:quantum-schur-circuit} illustrates a high-level representation of the circuit. QSVT is first applied to approximate $Q^{-1}$ by a polynomial transformation $P_{\mathrm{inv}}(x)\approx 1/x$, producing a block encoding $U_{Q^{-1}}$. 
This enables the construction of the Schur complement $S=\mathcal{A}Q^{-1}\mathcal{A}^T$ via 
$U_S \leftarrow U_{\mathcal{A}} U_{Q^{-1}} U_{\mathcal{A}}^T$, and the right-hand side 
$b=-r-\mathcal{A}Q^{-1}g$, prepared by successively applying $U_{Q^{-1}}$ to $U_g$, 
then $U_{\mathcal{A}}$, and finally combining with $U_r$ using the Linear Combination of Unitaries (LCU).  The system $S\lambda=b$ is then solved using a second QSVT inversion. 
Since $S=\mathcal{A}Q^{-1}\mathcal{A}^T$, its condition number can be bounded as 
$\kappa_S \leq \kappa_{\mathcal{A}}^2 \kappa_Q$, which suffices for selecting the polynomial degree in QSVT. 
The resulting block encoding $U_{S^{-1}}$ is applied to $U_b$ to obtain $U_\lambda \leftarrow U_{S^{-1}} U_b$, representing $\lambda=S^{-1}b$.  
Finally, the primal step $\Delta z = -Q^{-1}(g+\mathcal{A}^T\lambda)$ is assembled. 
The product $\mathcal{A}^T\lambda$ is formed by applying $U_{\mathcal{A}}^T$ to $U_\lambda$ and adding $U_g$ via LCU to produce $U_1$. 
Applying $U_{Q^{-1}}$ to $U_1$ yields $U_{\Delta z}$. 
Measurement of this encoding returns the classical vector $\Delta z$, which is passed to the SQP routine for line search and barrier updates. 
\begin{figure}[t]
    \centering
    \includegraphics[width=\linewidth]{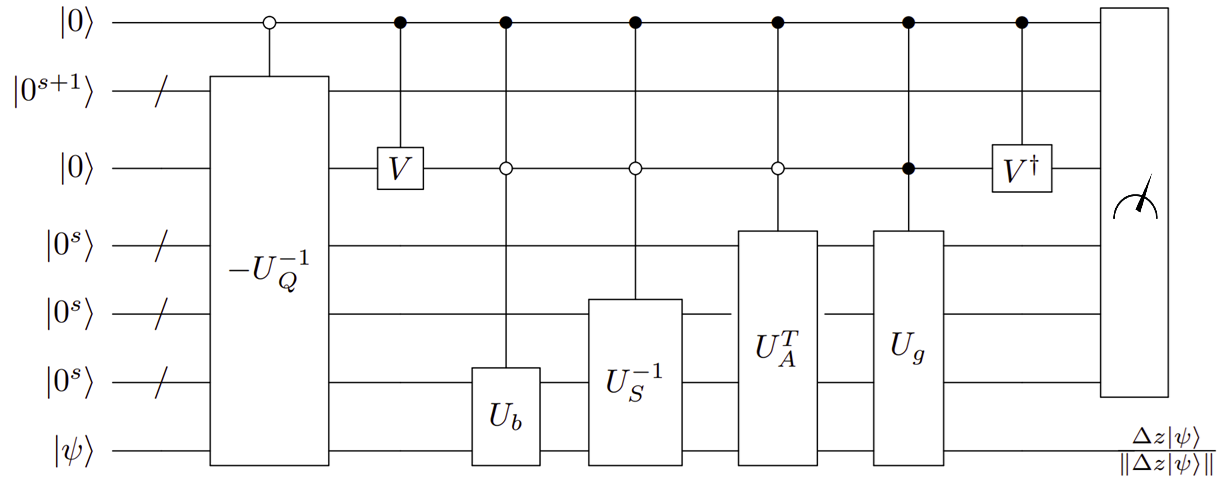}
    \caption{\small High level quantum circuit for calculating $\Delta z$. 
    }
    \label{fig:quantum-schur-circuit}
\vspace{-8mm}
\end{figure}

\begin{proposition}[Normalization Factor]
\label{prop:success-schur}
In Algorithm~\ref{Qschur}, the block encodings propagate such that the final block encoding of the step direction $\Delta z$ has normalization factor
$\alpha_{\Delta z}
  \;=\;
  \tfrac{\kappa_Q \beta_Q}{\alpha_Q}
  \left(
    \alpha_g
    + \alpha_{\mathcal{A}} \cdot \tfrac{\kappa_S \beta_S}{\alpha_{\mathcal{A}}^2 } \tfrac{\alpha_Q}{\kappa_Q \beta_Q}
      \big(
        \alpha_r + \alpha_{\mathcal{A}} \tfrac{\kappa_Q \beta_Q}{\alpha_Q} \alpha_g
      \big)
  \right)$.
Here $\beta_Q$ and $\beta_S$ are 
scaling parameters chosen in the QSVT inversion of $Q$ and $S$, ensuring that the polynomial approximation of $1/x$ remains bounded on $[-1,1]$ while accurately approximating inversion on the spectrum of the operator.
\end{proposition}
\begin{proof}
Tracing through Algorithm~\ref{Qschur}, the inversion of $Q$ by QSVT produces a block encoding with normalization $\alpha_{Q^{-1}} = \kappa_Q \beta_Q/ \alpha_Q$. Using this in the construction of the Schur Complement yields $\alpha_{S} =\alpha_{\mathcal{A}}^2 \alpha_{Q^{-1}}$. The right-hand side vector $b$ then has normalization $\alpha_b = \alpha_r + \alpha_{\mathcal{A}} \alpha_{Q^{-1}} \alpha_g$. The Schur complement inversion contributes $\alpha_{S^{-1}} = \kappa_S \beta_S / \alpha_S$, so the multiplier state $\lambda$ has normalization $\alpha_\lambda = \alpha_{S^{-1}} \alpha_b$. Forming $U_1$ results in $\alpha_1 = \alpha_g + \alpha_{\mathcal{A}} \alpha_\lambda$, and finally applying $Q^{-1}$ gives $\alpha_{\Delta z} = \alpha_{Q^{-1}} \alpha_1$. Substituting the intermediate expressions produces the closed-form normalization factor $\alpha_{\Delta z}$.
\end{proof}
\begin{remark}[Success Probability]
The normalization factor $\alpha_{\Delta z}$ determines the success probability 
of the algorithm. 
The final measurement succeeds with probability 
$P_{\text{succ}} = \frac{\|\Delta z\|^2}{\alpha_{\Delta z}^2}$,
implying that on average one needs
$1/{P_{\text{succ}} }$
repetitions of the algorithm to obtain the correct result.
Alternatively, amplitude amplification techniques can be employed to quadratically improve the success probability, reducing the expected number of repetitions \cite{brassard2000quantum,guerreschi2019repeat}.
\end{remark}

\begin{proposition}[Accuracy] \label{prop:accuracy}
Consider Algorithm \ref{Qschur}, the $(\alpha_Q,a_Q,\varepsilon_Q)$-block encoding of $Q$, 
$(\alpha_\mathcal{A},a_\mathcal{A},\varepsilon_\mathcal{A})$-block encoding of $\mathcal{A}$, 
$(\alpha_g,a_g,\varepsilon_g)$-block encoding of $g$, 
and $(\alpha_r,a_r,\varepsilon_r)$-block encoding of $r$,
the QSVT-based implementations of $Q^{-1}$ and $S^{-1}$ using polynomial approximations of accuracy 
$\varepsilon_{Q}^\prime$ and $\varepsilon_{S}^\prime$.
Then, there exist constants $C_1,\dots,C_6\ge0$, depending on the   normalization factors
such that the accuracy error  $\varepsilon_{\Delta z}$ satisfies the additive error bound
\begin{equation}\label{eq:final-additive}
\varepsilon_{\Delta z}
\;\le\;
C_1\,\varepsilon_Q 
+ C_2\,\varepsilon_\mathcal{A} 
+ C_3\,\varepsilon_g 
+ C_4\,\varepsilon_r 
+ C_5\,\varepsilon_{Q}^\prime
+ C_6\,\varepsilon_{S}^\prime .
\end{equation}
\end{proposition}
\begin{proof}
Tracing the error propagation through the main steps and using Remarks  
\ref{rem:add}-\ref{rem:QSVT} one can conclude that the error in $U_{Q^{-1}}$ satisfies
$\varepsilon_{Q^{-1}} \le C \varepsilon_Q + \varepsilon_{Q}^\prime$, for some $C>0$.
From the Schur complement $S$, multiplying $U_\mathcal{A}$, $U_{Q^{-1}}$, and $U_\mathcal{A}^T$ propagates errors $\varepsilon_S \le C_1 \varepsilon_\mathcal{A} + C_2 \varepsilon_{Q^{-1}}$ for some $C_1, C_2>0$. From the right-hand side $b$, $\varepsilon_b \le \varepsilon_r+C_1 \varepsilon_\mathcal{A} + C_2 \varepsilon_{Q^{-1}}+C_3\varepsilon_g$
for some $C_1, C_2, C_3>0$. The error in $U_{S^{-1}}$ satisfies
$\varepsilon_{S^{-1}} \le C \varepsilon_S + \varepsilon_{S}^\prime$. From computing $\lambda$, the error is $\varepsilon_\lambda \le \alpha_{S^{-1}} \varepsilon_{S^{-1}}+\alpha_b \varepsilon_b$. Computing $U_1$ gives the error $\varepsilon_1 \le \varepsilon_g+ \alpha_\mathcal{A}\varepsilon_\mathcal{A}+\alpha_\lambda\varepsilon_\lambda$. To compute $\Delta z$, the error satisfies 
$\varepsilon_{\Delta z} \le \alpha_{Q^{-1}} \varepsilon_{Q^{-1}}+\alpha_1 \varepsilon_1 $.
Collecting all terms, the total error is \eqref{eq:final-additive}
for suitable constants $C_1,\dots,C_6 \ge 0$ as claimed.
\end{proof}

\begin{proposition}[Complexity of the Quantum Schur Step]
\label{prop:schur-complexity}
Let $\varepsilon$ be a bound to all  block encoding tolerance errors, and
$\varepsilon^\prime$ the bound for the  polynomial approximations used in the QSVT.  
Assume that all matrices involved admit efficient block encodings using quantum-accessible
data structures, and that the conditions numbers $\kappa_Q$ and $\kappa_S$ of $Q$ and $S$ are bounded.
Algorithm \ref{Qschur} computes the primal update 
$\Delta z$ with a total time complexity
$\mathcal{O}\!\Big( \mathrm{polylog}(n_z/\varepsilon) \, \kappa\, \log(1/\varepsilon^\prime) \Big)$,
with $\kappa := \max\{\kappa_Q,\kappa_S\}$, 
$n_z$ 
the dimension of the decision vector.
\end{proposition}
\begin{proof}
Preparing block encodings of $Q, \mathcal{A}, g, r$ requires 
$\mathcal{O}(\mathrm{polylog}(n_z/\varepsilon))$ time \cite{gilyen2019quantum, lin2022lecture}. 
The inversion of $Q \in \mathbb{R}^{n_z \times n_z}$ via QSVT requires a polynomial of degree 
$\mathcal{O}(\kappa_Q \log(1/\varepsilon^\prime))$, resulting in a complexity of $\mathcal{O}(\mathrm{polylog}(n_z/\varepsilon)\kappa_Q \log(1/\varepsilon^\prime))$.
Forming the Schur complement 
$S\in \mathbb{R}^{(N+1)n \times (N+1)n}$ 
and the right-hand side $b$ uses block-encoding multiplications and additions, which contribute only polylogarithmic overhead. 
The inversion of $S$ via QSVT requires 
$\mathcal{O}(\mathrm{polylog}((N+1)n/\varepsilon)\kappa_S \log(1/\varepsilon^\prime))$, 
where $\kappa_S \leq \kappa(\mathcal{A})^2 \kappa(Q)$. 
The primal recovery step $\Delta z$ 
requires one further application of $Q^{-1}$ at the same cost as before. Thus, the runtime is dominated by the two QSVT inversions, leading to the mentioned result.
\end{proof}


\vspace{-4mm}

\section{Convergence and Complexity Analysis}
We now present an analysis of the convergence behavior and computational complexity of the proposed hybrid quantum–classical SQP algorithm, establishing both its robustness to perturbations and its potential for quantum speedup. To this end, we start by formalizing the algorithm as an iterative dynamical system and introducing key definitions for barrier-dependent and classical stationary points.

\begin{definition}[Algorithm \ref{alg:hybrid} as a Dynamical System]
Let \( z_i \in \mathbb{R}^{n_z} \) be the optimization iterate at iteration \( i \), \( \mu_i \in \mathbb{R}_{>0} \) the barrier parameter,
\( \varepsilon_{\Delta z} \in \mathbb{R}_{>0} \) the quantum solver error. Algorithm \ref{alg:hybrid} defines a iterative nonlinear dynamical system
\begin{equation} \label{eq:dynNLP}
z_{i+1} = \Phi(z_i, \mu_i, \varepsilon_{\Delta z}),   
\end{equation}
where \( \Phi \) is the update rule determined by Linearization \& Quadratic Approximation, Quantum Schur step, and Backtracking line search.
\end{definition}

\begin{definition}[Barrier-Dependent Stationary Point]
Let \( \mu_i > 0 \) be the barrier parameter at iteration \( i \). A point \( \bar{z}_i \in \mathbb{R}^{n_z} \) is a stationary point\footnote{Note that a stationary (KKT) point satisfies the first-order necessary conditions for optimality, but this alone does not ensure it is a local minimum; additional second-order sufficient conditions are required for local minimality such as the conditions stated in Assumption \ref{ass:regularity}.} of the barrier subproblem at \( \mu_i \) if it satisfies
$\nabla \tilde{F}( \bar{z}_i ; \mu_i ) + \nabla G( \bar{z}_i )^T \lambda_i = 0$,
subject to
$G( \bar{z}_i ) = 0, \quad H( \bar{z}_i ) < 0$,
where \( \tilde{F}(z; \mu) = F(z) + \mu \sum_{j=1}^{q} \phi\!\big( H_j(z) \big) \) is the barrier-augmented objective.
\end{definition}

\begin{definition}[Stationary Point of the Original Problem]
A point \( z^\star \in \mathbb{R}^{n_z} \) (possibly non unique) is called a \emph{stationary point} of the original constrained optimization problem if there exist multipliers \( \lambda^\star \in \mathbb{R}^{m_{\mathrm{eq}}} \) and \( \nu^\star \in \mathbb{R}^{m_{\mathrm{in}}} \) with \( \nu^\star \ge 0 \) such that
$\nabla F( z^\star ) + \nabla G( z^\star )^T \lambda^\star + \nabla H( z^\star )^T \nu^\star = 0$,
subject to
$G( z^\star ) = 0, \quad H( z^\star ) \le 0, \quad \nu^\star_j H_j( z^\star ) = 0 \quad \forall j$.
\end{definition}

Given the formal definition of Algorithm \ref{alg:hybrid} as a iterative nonlinear system \eqref{eq:dynNLP},
we are interested in the stability and robustness properties of the generated sequence \( \{z_i\} \).
Particularly, we analyze how the iterates behave under small perturbations in \( \mu_i \) and \( \varepsilon_{\Delta z} \), and how their limiting behavior depends on the magnitude of these perturbations. 
To this end, we use comparison functions from classes \( \mathcal{K} \) and \( \mathcal{KL} \) to characterize input-to-state stability (ISS) of the system \cite{jiang2001input}.

\begin{definition}[Local Input-to-State Stability]
The system \eqref{eq:dynNLP}
is said to be \emph{locally input-to-state stable (ISS)} around a stationary point \( z^\star \) if there exist \( \varepsilon, \delta > 0 \) and functions \( \beta \in \mathcal{KL} \), \( \gamma \in \mathcal{K}_\infty \) such that, for any sequence of perturbations with \( \|(\mu, \varepsilon_{\Delta z})\|_\infty := \max\{\sup_{i\in \mathbb{N}} \|\mu_i\|,\varepsilon_{\Delta z}\} < \delta \), the iterates satisfy
$\|z_i - z^\star\| \leq \beta(\|z_0 - z^\star\|, i) + \gamma(\|(\mu, \varepsilon_{\Delta z})\|_\infty)$,
for all \( i \in \mathbb{N} \), provided \( \|z_0 - z^\star\| < \varepsilon \).
\end{definition}

This property implies that the iterates \( z_i \) converge to a neighborhood of \( z^\star \), with the radius determined by the gain function \( \gamma \) and the size of the perturbations. If the perturbations vanish, the iterates converge to the stationary point \( z^\star \) of the original problem.
Furthermore, a sufficient condition for local ISS is the existence of a continuous, positive definite function \( V : \mathbb{R}^{n_z} \to \mathbb{R}_{\geq 0} \), constants \( \varepsilon > 0 \), \( \delta > 0 \), and functions \( \alpha, \gamma \in \mathcal{K}_\infty \) with \( \alpha < \mathrm{id} \), such that
$V(\Phi(z, \mu, \varepsilon_{\Delta z})) \leq \alpha(V(z)) + \gamma(\|(\mu, \varepsilon_{\Delta z})\|_\infty)$,
for all \( z \) with \( \|z - z^\star\| < \varepsilon \) and 
perturbations \( (\mu, \varepsilon_{\Delta z}) \) with \( \|(\mu, \varepsilon_{\Delta z})\|_\infty < \delta \).

We now state the main assumptions regarding problem regularity, properties of the barrier function, and the accuracy of the quantum Schur step, which underpin the subsequent convergence analysis.
\begin{assumption}[Problem Regularity]\label{ass:regularity}
Let $F:\mathbb{R}^{n_z} \to \mathbb{R}$, $G:\mathbb{R}^{n_z} \to \mathbb{R}^{m_{\mathrm{eq}}}$, and $H:\mathbb{R}^{n_z} \to \mathbb{R}^{m_{\mathrm{in}}}$ 
be twice continuously differentiable in a neighborhood of a stationary point $z^\star$ of the nonlinear program \eqref{eq:NLP}.
At $z^\star$, the Linear Independence Constraint Qualification (LICQ) holds; strict complementarity holds for the active inequalities; and the Second-Order Sufficient Condition (SOSC) holds.
Moreover, the QP
Hessian is the Gauss–Newton approximation possibly with damping
$Q \gets \nabla^2 F\!\left(z^{(i)}\right)
   + \mu \sum_{j=1}^{q} \phi''\!\big(H_j(z^{(i)})\big)\,\nabla H_j\!\left(z^{(i)}\right)\nabla H_j\!\left(z^{(i)}\right)^T + \sigma I$,
$\sigma\ge 0$, chosen so that $Q\succ 0$ on $\ker(\mathcal A)$.
\end{assumption}

\begin{assumption}[Barrier Function Properties]\label{ass:barrier}
The barrier function $\phi : (-\infty, 0) \to \mathbb{R}$ is strictly convex and twice continuously differentiable on its domain. Moreover, it satisfies $\lim_{s \to 0^-} \phi(s) = +\infty$, with $\phi'(s) < 0$ and $\phi''(s) > 0$ for all $s < 0$. An
example is the logarithmic barrier $\phi(s) = -\log(-s)$.
\end{assumption}

\begin{assumption}[Quantum Schur Step Accuracy]\label{ass:quantum}
At iteration $i$, the quantum subroutine
\[
\Delta z^{(i)} \;\gets\; \texttt{QuantumSchurStep}\!\big(Q^{(i)},\,\mathcal{A}^{(i)},\,g^{(i)},\,r^{(i)}\big)
\]
returns multipliers $\lambda^{(i)} \in \mathbb{R}^{m_{\mathrm{eq}}}$ and a primal step $\Delta z^{(i)} \in \mathbb{R}^{n_z}$ such that
$\|\Delta z^{(i)} - \Delta z^{(i)}_{\mathrm{Newton}}\| \;\le\; \varepsilon_{\Delta z}$,
where $\Delta z^{(i)}_{\mathrm{Newton}}$ is the exact\footnote{The ``exact Newton/SQP step'' refers to the solution $\Delta z^{(i)}_{\text{Newton}}$ of the quadratic subproblem arising from the second-order Taylor expansion of the barrier-augmented objective and first-order linearization of the constraints at the current iterate.} Newton/SQP step for the barrier subproblem at $z^{(i)}$, and $Q^{(i)} \in \mathbb{R}^{n_z \times n_z}$, $\mathcal{A}^{(i)} \in \mathbb{R}^{m_{\mathrm{eq}} \times n_z}$, $g^{(i)} \in \mathbb{R}^{n_z}$, $r^{(i)} \in \mathbb{R}^{m_{\mathrm{eq}}}$.
\end{assumption}
Note that Proposition \ref{prop:accuracy} guarantees that 
Assumption \ref{ass:quantum} holds.
The following useful lemma quantifies how closely the stationary point of the barrier subproblem approximates the true solution of the original constrained optimization problem as the barrier parameter decreases.
\begin{lemma}[Local Barrier Approximation Error]\label{lem:barrier-error}
Let $z^\star \in \mathbb{R}^{n_z}$ be a  stationary point of the original constrained optimization problem \eqref{eq:NLP}.  
Let $\bar{z}_i$ be a stationary point of the barrier subproblem at barrier parameter $\mu_i > 0$, 
$\min_{z} \Big[ F(z) + \mu_i \sum_{j=1}^q \phi\big( H_j(z) \big) \Big]
\quad \text{s.t. } G(z) = 0$,
where $\phi$ satisfies Assumption~\ref{ass:barrier}.  
Under Assumption~\ref{ass:regularity}, there exists $\bar{\mu} > 0$ such that for all $0 < \mu_i \le \bar{\mu}$ we have
$\|\bar{z}_i - z^\star\| \le \gamma(|\mu_i|)$,
for some class-$\mathcal{K}$ function $\gamma$, which can in fact be chosen in the linear form $\gamma(r) = C r$ for a constant $C > 0$ and all $0 < r \le \bar{\mu}$. Moreover, $\bar{z}_i \to z^\star$ as $\mu_i \to 0$.
\end{lemma}
\begin{proof}
The full proof can be found in \cite[Section~11.3]{nocedal2006numerical}.  
It mainly consists in analyzing the KKT conditions of the barrier subproblem:
$\nabla F(z) + \mu_i \sum_{j=1}^q \phi'(H_j(z)) \nabla H_j(z) + \nabla G(z)^T \lambda = 0$,  $G(z)=0$.
Define the mapping $\Psi(z,\lambda,\mu_i)$ that collects these conditions.  
At $\mu_i=0$, one recovers $\Psi(z^\star,\lambda^\star,0)=0$, which corresponds to the original KKT system.  
Under Assumption~\ref{ass:regularity}, the Jacobian of $\Psi$ with respect to $(z,\lambda)$ at $(z^\star,\lambda^\star,0)$ is nonsingular.  
The Implicit Function Theorem then guarantees the existence of a smooth local solution $(z(\mu_i),\lambda(\mu_i))$ for $\mu_i$ near $0$, with $z(0)=z^\star$.  
Differentiability further implies that there exists a constant $C > 0$ such that $\|\bar{z}_i - z^\star\| \leq C \mu_i$ for all sufficiently small $\mu_i$, which yields the stated bound.
\end{proof}

We can now establish the local ISS of the algorithm by showing that, under the stated assumptions, the iterates remain close to the true solution even in the presence of small perturbations.
\begin{theorem}[Local ISS of Algorithm \ref{alg:hybrid}]\label{thm:ISS-hybrid}
Consider the iterative system \eqref{eq:dynNLP}
generated by Algorithm \ref{alg:hybrid}.  
Suppose Assumptions~\ref{ass:regularity}, \ref{ass:barrier}, and \ref{ass:quantum} hold.
Then there exist constants $\varepsilon > 0$, $\delta > 0$, and functions $\beta \in \mathcal{KL}$, $\gamma \in \mathcal{K}$ such that, for any initial point $z_0$ with $\|z_0 - z^\star\| < \varepsilon$ and any perturbations satisfying
\( \|(\mu, \varepsilon_{\Delta z})\|_\infty := \max\{\sup_{i\in \mathbb{N}} \|\mu_i\|,\varepsilon_{\Delta z}\} < \delta \),
the iterates satisfy
\[
\|z_i - z^\star\| \le \beta(\|z_0 - z^\star\|, i) + \gamma\!\Big(\|(\mu, \varepsilon_{\Delta z})\|_\infty\Big), \quad \forall i \in \mathbb{N}.
\]
In particular, if $\mu_j \to 0$ and $\varepsilon_{\Delta z} \to 0$, then $z_i \to z^\star$ as $i \to \infty$.
\end{theorem}
\begin{proof}
The proof follows standard ISS arguments.  
Under Assumption~\ref{ass:regularity}, the exact Newton step for the barrier subproblem at $\mu_i$ satisfies the local contraction
$\| z_i + \alpha_i \Delta z_{\mathrm{Newton}}^{(i)} - \bar z_i \|
\;\le\; \rho \| z_i - \bar z_i \|$, $\rho \in (0,1),~\alpha_i \in (0,1]$.
Since the algorithm applies the inexact step $\Delta z^{(i)}$,
$z_{i+1} = z_i + \alpha_i \Delta z^{(i)}
        = z_i + \alpha_i \Delta z_{\mathrm{Newton}}^{(i)}
          + \alpha_i\!\big(\Delta z^{(i)}-\Delta z_{\mathrm{Newton}}^{(i)}\big)$.
By Assumption~\ref{ass:quantum} and $\alpha_i\le 1$, we obtain 
$\|z_{i+1} - \bar z_i\|= \|z_i + \alpha_i \Delta z^{(i)}_{\text{Newton}} - \bar z_i+ \alpha_i(\Delta z^{(i)} - \Delta z^{(i)}_{\text{Newton}})\|\leq \rho\,\|z_i - \bar z_i\| + \varepsilon_{\Delta z}$.
Applying the triangle inequality to $\|z_{i+1} - z^\star\|$, using Lemma \ref{lem:barrier-error}, and the previous result, we get
$\|z_{i+1} - z^\star\| \leq \|z_{i+1} - \bar z_i\| + \|\bar z_i - z^\star\| \le \rho\,\|z_i - \bar z_i\| + \varepsilon_{\Delta z} + C \mu_i$.
Since,
$\|z_i - \bar z_i\|
\leq \|z_i - z^\star\| + C \mu_i$.
Substituting this into the previous inequality yields
$\|z_{i+1} - z^\star\| \leq \rho\,\|z_i - z^\star\| + C' \mu_i + \varepsilon_{\Delta z} $
where $C'=\rho C+C$ is a constant.

Now, define the Lyapunov function $V(z_i) := \|z_i - z^\star\|^2$. From the previous bounds,
it follows that $V(z_{i+1})\leq \rho^2 V(z_i) + 2\rho \sqrt{V(z_i)} (C' \mu_i + \varepsilon_{\Delta z}) + (C' \mu_i + \varepsilon_{\Delta z})^2$.
Applying Young's inequality for any $\varepsilon > 0$,
$2\rho \sqrt{V(z_i)} (C' \mu_i + \varepsilon_{\Delta z}) \leq \varepsilon V(z_i) + \frac{\rho^2}{\varepsilon} (C' \mu_i + \varepsilon_{\Delta z})^2$.
Choosing $\varepsilon$ so that $\rho^2 + \varepsilon < 1$, we get
$V(z_{i+1}) \leq (\rho^2 + \varepsilon) V(z_i) + \left(1 + \frac{\rho^2}{\varepsilon}\right) (C' \mu_i + \varepsilon_{\Delta z})^2$.
Thus, 
we have
$V(z_{i+1}) \leq \alpha(V(z_i)) + \gamma(\|(\mu_i, \varepsilon_{\Delta z})\|_\infty)$,
from which one can conclude the claim.
\end{proof}

As a direct consequence, the following result  provides an explicit bound on the distance between the iterates and the true solution.
\begin{corollary}[Explicit Local Error Bound]\label{cor:explicit-bound}
Under the assumptions of Theorem~\ref{thm:ISS-hybrid}, there exist $\varepsilon > 0$, $\delta > 0$ and constants $C_1, C_2, C_3 > 0$ such that, for any initial point $z_0$ with $\|z_0 - z^\star\| < \varepsilon$ and any perturbations satisfying
\( \|(\mu, \varepsilon_{\Delta z})\|_\infty < \delta \), for all $i \in \mathbb{N}$,
$\|z_i - z^\star\| \leq C_1 \rho^i \|z_0 - z^\star\| + C_2 \sup_{k \leq i} \mu_k + C_3 \varepsilon_{\Delta z}$,
where $\rho \in (0,1)$ is the local contraction factor for the exact Newton step.
\end{corollary}

\begin{proof}
From the proof of Theorem~\ref{thm:ISS-hybrid}, we have the recursion
$\|z_{i+1} - z^\star\| \leq \rho \|z_i - z^\star\| + C_2 \mu_i + C_3 \varepsilon_{\Delta z}$,
where $\rho \in (0,1)$ and $C_2, C_3 > 0$ are constants.
We can unroll this recursion and conclude $\|z_i - z^\star\| \leq \rho^i \|z_0 - z^\star\| + C_2 \sum_{k=0}^{i-1} \rho^{i-1-k} \mu_k + C_3 \sum_{k=0}^{i-1} \rho^{i-1-k} \varepsilon_{\Delta z}$. From this, it follows the result.
\end{proof}

We can now establish a global convergence result, showing that under the stated assumptions, the algorithm generates bounded iterates that converge to a neighborhood of the set of KKT points of the original problem.
\begin{theorem}[Global Convergence to a Neighborhood of $z^\star$]
Consider the iterative system \eqref{eq:dynNLP} generated by Algorithm~\ref{alg:hybrid}.  
Suppose Assumptions~\ref{ass:regularity}--\ref{ass:quantum} hold, and let $\{\mu_i\}$ be a positive, non-increasing  sequence with $\mu_i \to \mu_{\min} \ge 0$. 
Assume further that the line search at each iteration enforces explicit decrease of equality infeasibility, i.e.,
$\|G(z_{i+1})\| \leq (1-\sigma\alpha_i)\|G(z_i)\|$, for some $\sigma \in (0,1)$,
where $\alpha_i$ is the accepted step size.

Then, every sequence $\{z_i\}$ generated by the algorithm is bounded. 
Moreover, for small enough $\mu_{\min}$ and $\varepsilon_{\Delta z}$, every sequence $\{z_i\}$ satisfies
\begin{equation}\label{eq:limsup}
\limsup_{i\to\infty}\,\mathrm{dist}(z_i,\mathcal{Z}^\star)\;\le\;C_1\,\varepsilon_{\Delta z} + C_2\,\mu_{\min},    
\end{equation}
where $\mathcal{Z}^\star$ is the set of KKT points (stationary points) of the original problem, and $C_1, C_2 > 0$ are constants depending only on problem regularity and line-search parameters.
\end{theorem}
\begin{proof}
To establish boundedness of the sequence $\{z_i\}$, observe that at each iteration, the algorithm solves a barrier-augmented subproblem with parameter $\mu_i$ and updates $z_{i+1}$ via a step that maintains strict feasibility with respect to the inequality constraints, i.e., $H(z_{i+1}) < 0$ for all $i$.
Given the barrier-augmented objective $\tilde F(z; \mu_i)$, it follows by construction that $z_i$ is strictly feasible, so $\tilde F(z_i; \mu_i)$ is finite for all $i$. Furthermore, the line search at each iteration enforces sufficient decrease of $\tilde F(z; \mu_i)$.
Since $\tilde F(z; \mu_i)$ is coercive on the strictly feasible set 
and since the sequence $\{\tilde F(z_i; \mu_i)\}$ is non-increasing, it follows that the sequence $\{z_i\}$ is contained in a sublevel set
$\mathcal{L}_0 := \left\{ z \in \mathbb{R}^d : H(z) < 0, \ \tilde F(z; \mu_0) \leq \tilde F(z_0; \mu_0) \right\}$,
which is compact by the coercivity of $\tilde F$ and the strict feasibility of $z_i$. Therefore, the sequence $\{z_i\}$ is bounded.
Note also that from the infeasibility decrease condition, it follows that $\|G(z_i)\|$ converges to zero as $i \to \infty$ because $\alpha_i \in (0,1]$ and $\sigma \in (0,1)$.
Now, we conclude \eqref{eq:limsup} by noting that for sufficiently small $\varepsilon_{\Delta z}$ and $ \mu_{\min}$, the error recursion established in the ISS analysis (see Theorem~\ref{thm:ISS-hybrid})
holds.
\end{proof}


The next result establishes the overall computational complexity of the proposed hybrid quantum-classical barrier SQP algorithm, showing it achieves exponential speedup in problem size compared to classical methods

\begin{theorem}[Time Complexity]
\label{prop:hybrid-sqp-complexity}
Consider that the barrier parameter is reduced geometrically according with $\mu^{(i+1)} =  \beta \mu^{(i)}$, $\beta\in(0,1)$). 
Then, Algorithm~\ref{alg:hybrid} 
returns the terminal iterate $z^{(K)}$ with time complexity 
$\mathcal{O}\!\Big(
   \log(1/\mu_{\min}) \cdot \mathrm{polylog}(n_z/\varepsilon) \, \kappa\, \log(1/\varepsilon^\prime)
\Big),$
where $\kappa := \max\{\kappa_Q,\kappa_S\}$, $n_z$ is the dimension of the decision vector, and $q$ is the number of inequality constraints.
\end{theorem}
\begin{proof}
Each outer iteration of Algorithm~\ref{alg:hybrid} performs:
(i) linearization and quadratic approximation (classical, polynomial in $n_z$ and dominated by the next step);
(ii) the Quantum Schur Step, which by Proposition~\ref{prop:schur-complexity}
costs $\mathcal{O}(\mathrm{polylog}(n_z/\varepsilon)\,\kappa\,\log(1/\varepsilon'))$; and
(iii) fraction-to-the-boundary initialization and backtracking (classical, negligible compared to (ii)).
Hence the per-iteration runtime is dominated by (ii), and the total runtime is $K$ times this cost.
Under the geometric update, $K=\big\lceil \log(\mu_0/\mu_{\min})/\lvert\log(1/\beta)\rvert \big\rceil$,
which yields the stated bound.
\end{proof}

\begin{figure}[t]
    \centering
    \includegraphics[width=\linewidth]{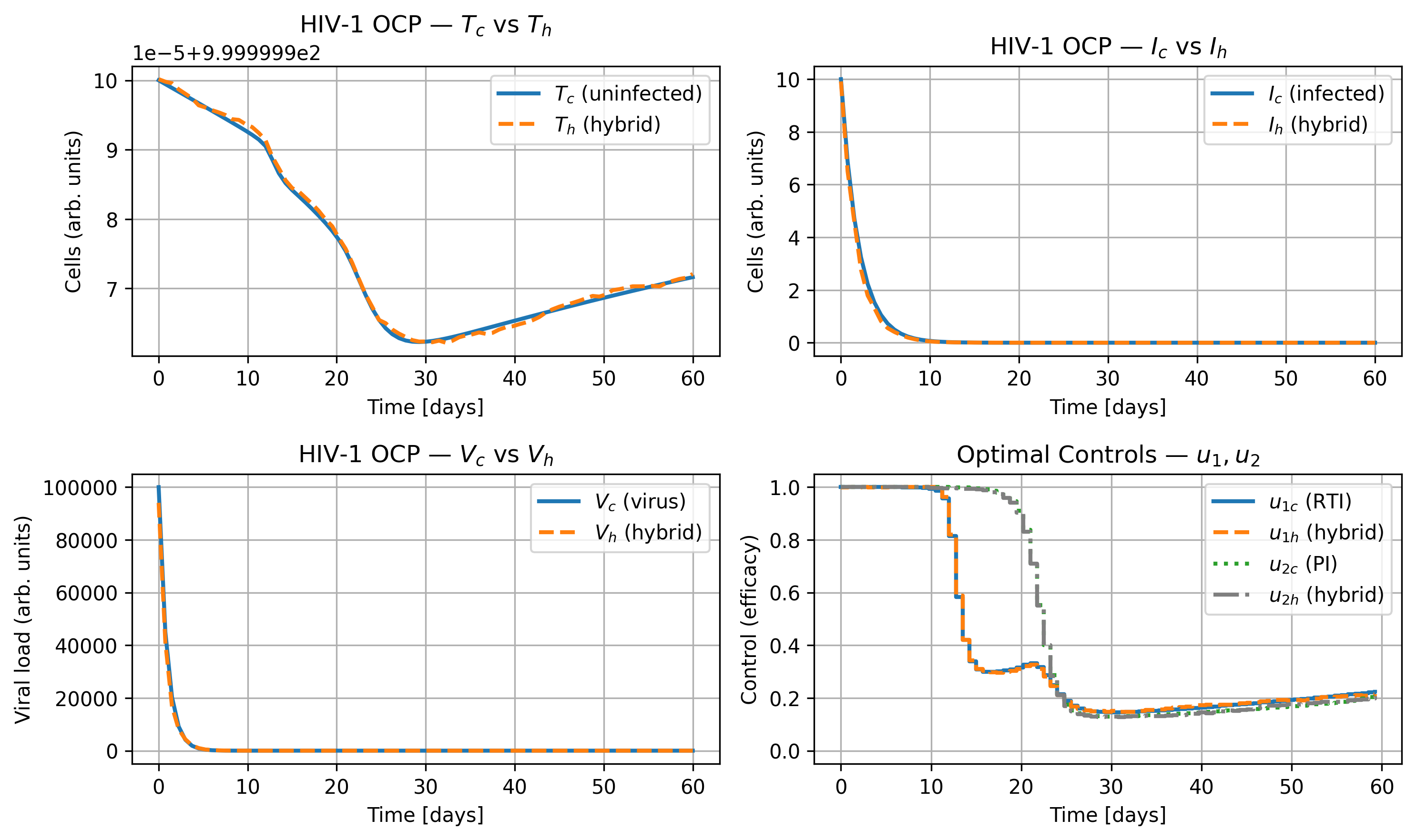}
    \caption{\small Comparison of classical SQP and hybrid quantum-classical SQP for the HIV-1 model. 
    }
    \label{fig:HIV_sim}
\vspace{-7mm}
\end{figure}

\section{Simulation Results}
We evaluate the proposed hybrid quantum-classical SQP algorithm on the Perelson HIV-1 model describing the interactions between uninfected CD4$^+$ T cells ($T$), infected cells ($I$), and free virus particles ($V$)
detailed in \cite{perelson1999}.
We introduce two drug controls: $u_{1,n}$: efficacy of reverse transcriptase inhibitors (RTIs), and $u_{2,n}$: efficacy of protease inhibitors (PIs). 
The objective is to minimize viral load and drug usage while maintaining healthy T cell levels. The cost function is
$J = \sum_{n=0}^{N-1} \left[ q_V V_n^2 + q_I I_n^2 + q_T (T_{\text{ref}} - T_n)^2 + r_1 u_{1,n}^2 + r_2 u_{2,n}^2 \right]+ q_V^f V_N^2 + q_I^f I_N^2 + q_T^f (T_{\text{ref}} - T_N)^2$,
where $q_V$, $q_I$, $q_T$ are state weights, $r_1$, $r_2$ are control weights, and $T_{\text{ref}}$ is the desired T cell level, the control bounds: $0 \le u_{1,n}, u_{2,n} \le 1$, and state bounds: $T_n, I_n, V_n \ge 0$. 
The horizon length was $N=60$ days with time step $\Delta t=1$. Initial states were chosen as $T_0=10^3$, $I_0=10$, $V_0=10^5$, and system parameters were taken from \cite{perelson1999}. 
Control bounds $0 \leq u_{1,n},u_{2,n} \leq 1$ and state positivity were enforced by logarithmic barriers.

Figure~\ref{fig:HIV_sim} shows the closed-loop trajectories obtained from the classical SQP solver and the hybrid solver. 
Both approaches converge to nearly identical solutions, with small deviations explained by the quantum Schur step tolerance~$\varepsilon_{\Delta z}$. While the classical SQP remains efficient for this small-scale problem, the hybrid method retains its convergence guarantees and offers favorable scaling for larger nonlinear optimal control problems.

\section{Conclusions}
In a classical SQP method, each iteration requires solving a Schur complement
system of dimension $n_z$, which typically incurs $\mathcal{O}(n_z^3)$ time.
By contrast,
the quantum Schur step replaces this cost with only polylogarithmic scaling
in $n_z$. However, it must be stressed that the QSVT step requires an estimate
of the condition number.
In our particular OCP structure, assuming that we set a positive definite weighting matrix $W$ for the stage cost $\ell(z_k)=\tfrac12\|r_k(z_k)\|_{W}^{2}$, 
this estimation is simplified because the
stage-wise Hessian blocks admit a Gauss–Newton form:
$Q_k \;=\; \big(\nabla_{z_k} r_k(z_k)\big)^T \, W \, \big(\nabla_{z_k} r_k(z_k)\big)
\;+\; B_k
\;+\; \sigma I$, where
 $B_k \ge 0$ is the barrier curvature term, and $\sigma I$ is a Levenberg–Marquardt damping term.
If $B_k$ is small and $W$ is well conditioned,
a practical approximation for the condition number is
$\kappa\!(Q_k) \;\approx\;
\tfrac{\|\nabla_{z_k} r_k(z_k)\|_2^{\,2} + \sigma}{\sigma}$. 


\vspace{-4mm}

\bibliographystyle{IEEEtran}
\bibliography{IEEEabrv,ACC}
\end{document}